%% file: main.tex
\documentclass{llncs}
\usepackage[utf8]{inputenc}
\usepackage[T1]{fontenc}
\usepackage[english]{babel}

\usepackage{amsmath}
\usepackage{amssymb}
\usepackage{amsfonts}
\usepackage{stmaryrd}

\usepackage[english, onelanguage, linesnumbered, lined, algoruled]{algorithm2e}
\SetKwComment{Comment}{//}{}
\usepackage{xspace} 

\newcommand{\Z}{\mathbb{Z}}

\newcommand{\ZqZ}{\Z/q\Z}

\newcommand{\eqdef}{\stackrel{\text{def}}{=}}

\newcommand{\Ceil}[1]{\left\lceil #1 \right\rceil}
\newcommand{\Ceiling}[1]{\Ceil{#1}}

\newcommand{\F}{\mathbb{F}_2}

\newcommand{\Fq}{\mathbb{F}_q}

\newcommand{\Fqm}{\mathbb{F}_{q^m}}

\newcommand{\CG}[2]{\begin{bmatrix}#1 \\ #2\end{bmatrix}_q}
\newcommand{\cg}[2]{\CG{#2}{#1}}

\newcommand{\C}{{\mathcal{C}}}

\newcommand{\rand}{\stackrel{\$}{\leftarrow}}

\newcommand{\word}[1]{\ensuremath{\boldsymbol{#1}}}
\newcommand{\Av}{\word{A}}
\newcommand{\Bv}{\word{B}}
\newcommand{\Cv}{\word{C}}

\newcommand{\Gv}{\word{G}}
\newcommand{\Hv}{\word{H}}
\newcommand{\Iv}{\word{I}}

\newcommand{\Mv}{\word{M}}

\newcommand{\Sv}{\word{S}}
\newcommand{\Tv}{\word{T}}
\newcommand{\Uv}{\word{U}}
\newcommand{\Vv}{\word{V}}

\newcommand{\Xv}{\word{X}}
\newcommand{\Yv}{\word{Y}}

\newcommand{\cv}{\word{c}}

\newcommand{\ev}{\word{e}}

\newcommand{\gv}{\word{g}}
\newcommand{\hv}{\word{h}}

\newcommand{\sv}{\word{s}}

\newcommand{\uv}{\word{u}}
\newcommand{\vv}{\word{v}}
\newcommand{\wv}{\word{w}}
\newcommand{\xv}{\word{x}}
\newcommand{\yv}{\word{y}}

\newcommand{\sigmav}{\word{\sigma}}

\DeclareMathOperator{\Rank}{Rank}
\DeclareMathOperator{\rank}{Rank}
\DeclareMathOperator{\Supp}{Supp}
\DeclareMathOperator{\Prob}{\mathbb P}

\newcommand{\norme}[1]{\| #1 \|}

\newcommand{\RSD}{\textsf{RSD}\xspace}
\newcommand{\DRSD}{\textsf{DRSD}\xspace}
\newcommand{\IRSD}{\textsf{IRSD}\xspace}
\newcommand{\DIRSD}{\textsf{DIRSD}\xspace}
\newcommand{\RSR}{\textsf{RSR}\xspace}
\newcommand{\xRSR}{\textsf{xRSR}\xspace}
\newcommand{\RSL}{\textsf{RSL}\xspace}
\newcommand{\DRSL}{\textsf{DRSL}\xspace}
\newcommand{\DIRSL}{\textsf{DIRSL}\xspace}
\newcommand{\IRSL}{\textsf{IRSL}\xspace}
\newcommand{\LRPC}{\textsf{LRPC}\xspace}
\newcommand{\ILRPC}{\textsf{ILRPC}\xspace}

\newcommand{\Adv}{\ensuremath{\mathsf{Adv}}}
\newcommand{\mA}{\ensuremath{\mathcal{A}}}

\usepackage{booktabs}
\usepackage{dsfont} 
\usepackage{multirow}
\usepackage{graphicx}
\usepackage{placeins}
\usepackage[pagebackref]{hyperref}
\newcommand{\mapolicebackref}[1]{\mbox{\textsl{\small #1}}}
\renewcommand*{\backref}[1]{}
\renewcommand*{\backrefalt}[4]{%
\ifcase #1 \mapolicebackref{Uncited in this paper}
    \or \mapolicebackref{#2}
    \else \mapolicebackref{#2}
\fi
}

\usepackage{marvosym}
\usepackage{wasysym}
\usepackage{pdfpages}

\newcommand{\KEM}{\mathsf{KEM}}
\newcommand{\KeyGen}{\normalfont\textsf{KeyGen}}
\newcommand{\INDCPA}{\mathsf{IND\mbox{-}CPA}}
\newcommand{\Encap}{\normalfont\textsf{Encap}}
\newcommand{\Decap}{\normalfont\textsf{Decap}}
\newcommand{\INDCPArand}{\INDCPA_{\mathsf{rand}}}
\newcommand{\INDCPAreal}{\INDCPA_{\mathsf{real}}}
\newcommand{\advA}{\mathcal{A}} 
\newcommand{\AdvINDCPA}[2]{\Adv^{\normalfont{indcpa}}_{#1}(#2)}
\newcommand{\Exp}{\mathbf{Exp}}
\newcommand{\param}{\ensuremath{\mathsf{param}}}
\newcommand{\E}{\ensuremath{\mathcal{E}}\xspace}
\newcommand{\sk}{\ensuremath{\mathsf{sk}}\xspace}
\newcommand{\pk}{\ensuremath{\mathsf{pk}}\xspace}
\newcommand{\ct}{\ensuremath{\mathsf{ct}}\xspace}
\newcommand{\A}{\ensuremath{\mathcal{A}}\xspace}
\newcommand{\ind}{\normalfont\textsf{ind}}
\newcommand{\seck}{{\lambda}}
\newcommand{\sets}{\leftarrow}
\newcommand{\Setup}{\ensuremath{\mathsf{Setup}}}
\newcommand{\comreturn}{\texttt{RETURN\ }}

\newcommand{\GUESS}{\texttt{GUESS}}

\newcommand{\Expec}{\mathbb E}

\newcommand{\vect}[1]{\langle #1 \rangle}
\newcommand{\EF}{EF}

\title{LRPC codes with multiple syndromes: near ideal-size KEMs without ideals}
\author{
Carlos Aguilar-Melchor\inst{1} \and
Nicolas Aragon\inst{2} \and
Victor Dyseryn \inst{3} \and
Philippe Gaborit \inst{3} \and
Gilles Zémor \inst{4}}

\institute{
Sandbox AQ. \and
CNRS, Inria, IRISA, Université de Rennes, France. \and
XLIM, Université de Limoges, France. \and
Institut de Mathématiques de Bordeaux, UMR 5251, France.}
\date{}

\begin{document}
\maketitle

\begin{abstract}
We introduce a new rank-based key encapsulation mechanism (KEM) with public key and ciphertext sizes around 3.5 Kbytes each, for 128 bits of security, without using ideal structures. Such structures allow to compress objects, but give reductions to specific problems
whose security is potentially weaker than for unstructured problems.
To the best of our knowledge, our scheme improves in size all the existing unstructured post-quantum lattice or code-based algorithms such as
FrodoKEM or Classic McEliece. Our technique, whose efficiency relies on properties of rank metric, is to build upon existing
Low Rank Parity Check (LRPC) code-based KEMs and to send multiple syndromes in one ciphertext, allowing to reduce the parameters
and still obtain an acceptable decoding failure rate. 
Our system relies on the hardness of the Rank Support Learning problem, a well-known variant of the Rank Syndrome Decoding problem.
The gain on parameters is enough to significantly close the gap between ideal and non-ideal constructions. It enables to choose an error weight close to the rank Gilbert-Varshamov bound, which is a relatively harder zone for algebraic attacks.
We also give a version of our KEM that keeps an ideal structure and permits to roughly divide the bandwidth by two compared to previous versions of LRPC KEMs submitted to the NIST with a Decoding Failure Rate (DFR) of $2^{-128}$.
\end{abstract}


\section{Introduction and previous work}
In recent years and especially since the 2017 NIST call for proposals on post-quantum cryptography, there has been a burst of activity in this field. Recent publications, such as the Barbulescu \textit{et al.} attack against the small characteristic discrete logarithm problem~\cite{BGJT14}, stress the importance of having a wide diversity of cryptographic systems before the emergence of large and fault-tolerant quantum computers.

The most common algorithms in post-quantum cryptography are lattice-based or code-based. Code-based cryptography relies on difficult
problems related to error-correcting codes embedded in Hamming metric spaces (often over small fields $\Fq$). Lattice-based cryptography is mainly based on the study of $q$-ary lattices,
which can be seen as codes over rings of type $\ZqZ$ (for large $q$), embedded in Euclidean metric spaces.

In this paper we study a rank-based cryptosystem. Rank-based cryptography is similar to code-based cryptography, with the difference that the error-correcting codes are embedded in a rank-metric space (often over a prime order field extension).

In rank metric, the practical difficulty of usual decoding problems grows very quickly with parameter size, which makes it very appealing for cryptography. This metric was introduced by Delsarte and Gabidulin~\cite{G85}, along with Gabidulin codes which are a rank-metric equivalent of Reed-Solomon codes.
Since then, rank-metric codes have been used for multiple applications such as coding theory and cryptography.

Among the different cryptographic primitives, rank-based cryptography literature is mainly focused around encryption schemes, even if rank metric is relevant to produce small size and general purpose digital signatures, such as Durandal~\cite{aragon2019durandal}. 
Until recently, the main approach to build cryptosystems based on rank-metric decoding problems was masking Gabidulin codes \cite{GPT91} in different ways and using the McEliece (or Niederreiter) setting with these codes.
Most cryptosystems based on this idea were broken by attacks which exploit the particular structure of Gabidulin codes (\cite{O08}, \cite{FL05}, \cite{BL04}, \cite{L11},\cite{G08}). 
A similar situation exists in the Hamming case for which most cryptosystems based on Reed-Solomon codes have been broken for a similar reason: these codes are so structured that they are difficult to mask
and there is always some structural information leak. 

To solve this difficulty, rank-based cryptosystems designers have either produced schemes without masking \cite{AABBBDGZ17} or proposed to use LRPC codes which are easier to mask. The latter are the foundations of the new cryptosystem presented in this paper.

LRPC codes were first introduced in \cite{gaborit2013low} and are a family of rank-metric error-correcting codes which admit a parity matrix whose coordinates generate a low rank vector space. They have a strong decoding power, and can be seen as the rank-metric equivalent to Hamming-metric MDPC codes, which are for example featured in third-round NIST candidate BIKE \cite{BIKE}.

In their quasi-cyclic form, LRPC codes are the main building block of the second-round NIST candidate ROLLO~\cite{ABDGHRTZABBBO19}.
This candidate was not selected for the third round due to algebraic attacks \cite{bardet2020algebraic}, \cite{bardet2020improvements}, which
significantly reduced the security of the parameters proposed in the original submissions. NIST encouraged further study on rank-metric cryptosystems \cite{alagic2020status}, but these attacks have not been improved. The NIST standardization process has improved the scientific community understanding of LRPC cryptosystems and of the associated attacks, meanwhile there were two points which still could be improved.

A first point was related to constant time implementations which were unsatisfactory and for which a recent paper \cite{chou2022constant} showed how it was possible to drastically improve their performances. A second point was the Decoding Failure Rate (DFR). Indeed, LRPC parameters led to quite efficient cryptosystems for DFRs around $2^{-30}$, but for DFRs below $2^{-128}$ there was a significant efficiency drop, as to obtain such DFRs the codes needed to be quite long.
The present paper shows how to avoid larger code lengths while still obtaining very low DFRs. This results in a significant improvement of the associated cryptosystems, both for the structured and unstructured case, without compromising a precise analysis of the DFR. 

A usual technique to reduce public key and ciphertext sizes in cryptosystems is to introduce structure in the underlying algebraic objects. This is done in general by introducing some extra ideal, module, or ring structure \cite{BIKE}, \cite{AABBBDGZ17}, \cite{HPS98}, \cite{Kyber}. However, adding structure comes at the cost of losing reductions to difficult problems in the more general form. A hypothesis must be made: that the structure does not set the stage for better attacks than for the unstructured, general, problem. 

When compared to structured finalist and alternate candidates to the NIST PQC standardization, using the standard communication metric (public key + ciphertext size), our cyclic scheme is more efficient than BIKE~\cite{BIKE} and HQC~\cite{AABBBDGPZ21a} (at least 1.4 times shorter for 128 bits of security), but somewhat less than structured lattice approaches (roughly 1.4 times larger for 128 bits of security) and significantly less than SIKE.

Using again the communication metric, our schemes perform very well in the unstructured setting. Among the finalist and alternate candidates, only two candidates do not use any ideal-like structure: FrodoKEM~\cite{FrodoKEM} for lattice-based cryptography and Classic McEliece~\cite{ClassicMcEliece} for code-based cryptography. A proposal of an unstructured code-based KEM was also introduced after the beginning of the NIST project, called Loong.CCAKEM~\cite{wang2019loong}. Our non-cyclic proposal compares advantageously to the three of them (2.8 times shorter than FrodoKEM for 128 bits of security).\\

\noindent{\bf Description of our technique.}
The usual approach to build a cryptosystem based on LRPC codes is to send a syndrome $\sv = \Hv\ev$ as a ciphertext where both $\ev$ and $\Hv$ are of low rank weight and, to decrypt, use a Rank Support Recovery algorithm to recover the support of $\ev$ from the support of the product $\Hv\ev$. The main obstacle to reduce the parameters for such cryptosystems is not the threat of cryptanalytic attacks but the DFR. In order for the Rank Support Recovery algorithm to work, the syndrome $\sv$ has to be large enough so as to generate the full product of supports of $\Hv$ and $\ev$.

As~\cite{GHPT17a_sv} and~\cite{wang2019loong} did, we use multiple syndromes $\sv_1,\dots,\sv_{\ell}$ of same error support. Our main result is to show that this approach, in the context of LRPC codes, can solve decryption failure rate issues that affected previous schemes and thus reduce significantly key and ciphertext size, as the Rank Support Recovery algorithm gets more coordinates to recover the product of supports. While intuitive at a first glance, the proof that multiple syndromes reduce the probability of failure is quite technical and led us to formulate a general result on the product of two random homogeneous matrices.

Sending multiple syndromes leads naturally to reducing the security of our KEM to the Rank Support Learning ($\RSL$) problem \cite{GHPT17a_sv,wang2019loong,aragon2019durandal}, which implies that our approach is specific to rank-metric and cannot be used for Hamming-based cryptosystems. Indeed, in a Hamming metric context, the complexity of the Support Learning problem decreases way faster with the number of given syndromes than its rank-metric counterpart, and thus---with a direct application of our approach---it is not possible in Hamming metric to obtain parameter sets that have a practical interest and are secure.\\ 

\noindent{\bf Contributions of the paper.}
We present in this paper five contributions:
\begin{itemize}
	\item A new LRPC code-based key encapsulation mechanism built upon the multiple syndrome approach that significantly improves decoding. We give an unstructured version of our KEM that achieves a competitive size of around 3.5 Kbytes each for the public key and ciphertext. We also give an ideal version to reduce the sizes even further.
	\item A proof that with our new approach, small weight parameters $r$ and $d$ of the LRPC code can be chosen higher and even very close to the rank Gilbert-Varshamov bound $d_{RGV} = \mathcal O(n)$, whereas for a LRPC code-based cryptosystem these values have to be in $\mathcal O(\sqrt n)$. When target weights increase, algebraic attacks become less effective and can even be more costly than combinatorial attacks.
	\item A probabilistic result on the support generated by the coordinates of a product matrix $\Uv\Vv$ where $\Uv$ and $\Vv$ are two random homogeneous matrices of low weight. This result happens to be the cornerstone of the efficiency of our KEM but is also general enough to be applicable elsewhere in cryptography or in other fields.
	\item A solution to reduce the failure rate of the Rank Support Recovery algorithm when the dimension $m$ of the ambient space is low. This allows us to suggest a slightly modified version of our KEM that achieves even lower sizes for public key and ciphertext.
	\item An application of the multiple-syndrome approach to another existing rank-metric KEM: Ouroboros \cite{DGZ17}.
\end{itemize}
 
\medskip

\noindent{\bf Organization of the paper.}
The paper is organized as follows: Section 2 recalls basic facts about the rank-metric and the corresponding difficult problems,
Section 3 gives a background on LRPC codes and their decoding,
Section 4 introduces a new KEM using the multiple-syndrome technique to decode LRPC codes,
Section 5 is dedicated to proving a probabilistic result on the support of a product of homogeneous matrices, necessary to prove the efficiency of the new KEM,
Section 6 proves the IND-CPA property of the KEM,
Section 7 is concerned with parameters and performances of our KEMs and
finally Section 8 generalizes the multiple-syndrome approach to Ouroboros. 

\input{rank_metric.tex}

\input{lrpc.tex}

\input{multi-lrpc.tex}

\input{proof_multi_rsr.tex}

\input{security.tex}

\input{parameters.tex}

\input{generalization_ouroboros.tex}

\section{Conclusion and future work}

We provided a proof that, using multiple syndromes on rank-metric key encapsulation mechanisms, it is possible to obtain unexpectedly low decoding failure rates with efficient parameters. As a result, it is possible to obtain KEMs with small ciphertext and public key sizes even without ideal structure. We provide an IND-CPA proof for our scheme, whose security relies on the hardness of the $\DRSL$ and $\LRPC$ distinguishing problems. We give a quick description of the application of our approach to another rank-based cryptosystem that does not assume the hardness of the $\LRPC$ distinguishing problem, and only relies on $\DRSL$. A possible future work could be to provide a state of the art implementation of this scheme both in software and hardware.

\bibliographystyle{splncs04}
\bibliography{codecrypto,biblio}

\end{document}

%% file: rank_metric.tex
\section{Background on Rank Metric Codes}
\label{sec:RankMetric}
\subsection{General definitions}


Let $\Fq$ denote the finite field of $q$ elements where $q$ is the power of a prime and let $\Fqm$ denote the field of $q^m$ elements seen as the extension of degree $m$ of $\Fq$.

$\Fqm$ is also an $\Fq$ vector space of dimension $m$, we denote by capital letters the $\Fq$-subspaces of $\Fqm$ and by lower-case letters the elements of $\Fqm$.


We denote by $\vect{x_1,\dots, x_n}$ the $\Fq$-subspace generated by the elements $(x_1,\dots,x_n) \in \Fqm^n$.

Vectors are denoted by bold lower-case letters and matrices by bold capital letters (eg $\xv = (x_1,\dots, x_n) \in \Fqm^n$ and $\Mv = (m_{ij})_{\substack{1\leqslant i \leqslant k\\1\leqslant j \leqslant n}} \in \Fqm^{k\times n}$).

Let $P \in \Fq[X]$ be a polynomial of degree $n$. We can identify the vector space $\Fqm^n$ with the ring $\Fqm[X]/\langle P \rangle$, by mapping $\vv = (v_0,\dots, v_{n-1})$ to $\Psi(\vv) = \sum_{i=0}^{n-1} v_i X^i$. For $\uv, \vv \in \Fqm^n$, we define their product similarly as in $\Fqm[X]/\langle P \rangle$: $\wv = \uv\vv \in \Fqm^n$ is the only vector such that $\Psi(\wv) = \Psi(\uv)\Psi(\vv) \mod P$. In order to lighten the formula, we will omit the symbol $\Psi$ in the future.

If $S$ is a finite set, we denote by $x\rand S$ the fact that $x$ is chosen uniformly at random amongst $S$.

The number of $\Fq$-subspaces of dimension $r$ of $\Fqm$ is given by the Gaussian coefficient 
\[
\cg{r}{m} = \prod_{i=0}^{r-1} \frac{q^m-q^i}{q^r-q^i}.
\]

\begin{definition}[Rank metric over $\Fqm^n$] \label{def:RankMetric}
Let $\xv=(x_1,\dots,x_n) \in \Fqm^n$ and let $(b_1,\dots ,b_m) \in \Fqm^m$ be a basis of $\Fqm$ over $\Fq$.
Each coordinate $x_j$ is associated to a vector of $\Fq^m$ in this basis: $x_j = \sum_{i=1}^m m_{ij} b_i$. The $m \times n$ matrix
associated to $\xv$ is given by $\Mv(\xv)=(m_{ij})_{\substack{1
    \leqslant i \leqslant m \\ 1 \leqslant j \leqslant n}}$.

The rank weight $\norme{\xv}$ of $\xv$ is defined as 
\[
\norme{\xv} \eqdef \Rank \Mv(\xv).
\]
This definition does not depend on the choice of the basis.
The associated distance $d(\xv,\yv)$ between elements $\xv$ and $\yv$ in $\Fqm^n$ is defined by 
$d(\xv,\yv)=\norme{\xv-\yv}$.

The support of $\xv$, denoted $\Supp(\xv)$, is the $\Fq$-subspace of $\Fqm$ generated by the coordinates of $\xv$:
\[
\Supp(\xv) \eqdef \langle x_1, \dots, x_n\rangle
\]
and we have $\dim \Supp(\xv) = \norme{\xv}$.
\end{definition}

\begin{definition}[$\Fqm$-linear code]\label{def:FqmLinearCode}
An $\Fqm$-linear code $\C$ of dimension $k$ and length $n$ is a
subspace of dimension $k$ of $\Fqm^n$ seen as a rank metric space. The
notation $[n,k]_{q^m}$ is used to denote its parameters.

The code $\C$ can be represented by two equivalent ways:
\begin{itemize}
\item by a generator matrix $\Gv \in \Fqm^{k\times n}$. Each row of $\Gv$ is an element of a basis of $\C$,
\[
\C = \{\xv\Gv, \xv \in \Fqm^k \}.
\]
\item by a parity-check matrix $\Hv \in \Fqm^{(n-k)\times n}$. Each row of $\Hv$ determines a parity-check equation verified by the elements of $\C$:
\[
\C = \{\xv \in \Fqm^n : \Hv\xv^T = \boldsymbol{0} \}.
\]
\end{itemize}
We say that $\Gv$ (respectively $\Hv$) is under systematic form if and only if it is of the form $(\Iv_k|\Av)$ (respectively $(\Iv_{n-k}|\Bv)$).
\end{definition}


\subsection{Ideal codes}
To describe an $[n,k]_{q^m}$ linear code, we can give a systematic generator matrix or a systematic parity-check matrix. In both cases, the number of bits needed to represent such a matrix is $k(n-k)m\Ceil{\log_2 q}$.  To reduce the size of a representation of a code, we introduce ideal codes. They are a generalization of double circulant codes by choosing a polynomial $P$ to define the quotient-ring $\Fqm[X]/(P)$. More details about this construction can be found in \cite{aragon2019low}.

\begin{definition}[Ideal codes]\label{def:IdealCodes}
Let $P(X) \in \Fq[X]$ be a polynomial of degree $n$ and $\gv_1,\gv_2
\in \Fqm^k$. Let $G_1(X) = \sum_{i=0}^{k-1} g_{1i}X^i$ and $G_2(X) =
\sum_{j=0}^{k-1} g_{2j}X^j$ be the polynomials associated respectively to $\gv_1$ and $\gv_2$.
We call the $[2k,k]_{q^m}$ {\em ideal code $\C$ of generator $(\gv_1,\gv_2)$} the code with generator matrix

$$
\Gv = \begin{pmatrix}
G_1(X) \mod P			& \vline &G_2(X) \mod P \\
XG_1(X) \mod P 			& \vline &XG_2(X) \mod P \\
	\vdots				& \vline &\vdots \\
X^{k-1}G_1(X) \mod P 	& \vline &X^{k-1}G_2(X) \mod P 
\end{pmatrix}.
$$

More concisely, we have 
$\C = \{ (\xv\gv_1 \mod P, \xv\gv_2 \mod P), \xv\in \Fqm^k \}$.
We will often omit mentioning the polynomial $P$ if there is no ambiguity.

We usually require $\gv_1$ to be invertible, in which case the code
admits the systematic form, $\C = \{(\xv,\xv\gv), \xv\in \Fqm^k \}$ with $\gv = \gv_1^{-1}\gv_2 \mod P$. 
\end{definition}

%
%

\subsection{Difficult problems in rank metric}

\subsubsection{Rank Syndrome Decoding and ideal variant}


\begin{problem}[Rank Syndrome Decoding]
\label{prob:RSD}
On input $(\Hv, \sv) \in \Fqm^{(n-k)\times n}\times\Fqm^{(n-k)}$, the Rank Syndrome Decoding Problem $\RSD_{n, k, r}$ is to compute $\ev \in \Fqm^n$ such that $\Hv\ev^\intercal = \sv^\intercal$ and $\norme{\ev} = r$. 
\end{problem}

In \cite{GZ14} it is proven that the Syndrome Decoding problem in the Hamming metric, which is a well-known NP-hard problem, is probabilistically reduced to the $\RSD$ problem. Moreover, the $\RSD$ problem can be seen as a structured version of the NP-Hard  MinRank problem~\cite{buss1999computational}, indeed the MinRank problem is equivalent to the $\RSD$ problem replacing $\Fqm$-linear codes by $\Fq$-linear codes.
The variant of this problem for ideal codes is as follows.

%


\begin{problem}[Ideal-Rank Syndrome Decoding]\label{prob:I-RSD}
Let $P\in \Fq[X]$ a polynomial of degree $k$. On input $(\hv, \sigmav) \in \Fqm^k\times\Fqm^{k}$, the Ideal-Rank Syndrome Decoding Problem $\IRSD_{2k, k, r}$ is to compute $\xv = (\xv_1,\xv_2) \in \Fqm^{2k}$ such that $\xv_1 + \xv_2\hv = \sigmav \mod P$ and $\norme{\xv} = r$. 
\end{problem}

Since $\hv$ and $P$ define a systematic parity-check matrix of a
$[2k,k]_{q^m}$ ideal code, the $\IRSD$ problem is a particular case of
the $\RSD$ problem. Although this problem is theoretically easier than the 
$\RSD$ problem, in practice the best algorithms for solving both these problems are the same.

\subsubsection{Rank Support Learning}

The following problem was introduced in \cite{GHPT17a_sv}. It is similar
to the $\RSD$  problem, the difference is that instead of having one syndrome, we are given several syndromes of errors of same support and the goal is to find this support. The security of $\RSL$ is considered to be similar to $\RSD$ for a small number of syndromes. More details about the security of $\RSL$ are provided in Section \ref{sec:secu}.


%
%

\begin{problem}{\textbf{Rank Support Learning (RSL)}} \cite{GHPT17a_sv}
\label{prob:RSL}
On input $(\Hv, \Sv) \in \Fqm^{(n-k)\times n}\times\Fqm^{\ell\times (n-k)}$, the Rank Support Learning Problem $\RSL_{n, k, r, \ell}$ is to compute a subspace  $E$ of $\Fqm$ of dimension $r$, such that there exists a matrix $\Vv \in E^{\ell\times n}$ such that $\Hv\Vv^\intercal = \Sv^\intercal$
\end{problem}

The $\RSL$ problem also has an ideal variant called $\IRSL$.

\subsubsection{Decisional problems}

For all the problems $\RSD, \IRSD, \RSL$ and $\IRSL$ defined above, we can give a decisional version whose goal is to distinguish (for the example of $\RSD$) between a random input $(\Hv, \sv)$ or an actual syndrome input $(\Hv, \Hv\ev^\intercal)$.  We denote these decisional versions $\DRSD, \DIRSD, \DRSL$ and $\DIRSL$. The reader is referred to \cite{aragon2019low} for more details about decisional problems.

%% file: lrpc.tex
\section{LRPC codes and their decoding} \label{sec:LRPC}

\subsection{Low Rank Parity Check codes}

LRPC codes were introduced in \cite{gaborit2013low}. They are the equivalent of MDPC codes from the Hamming metric. They have a strong decoding power and a weak algebraic structure, therefore they are well suited codes for cryptography.
\begin{definition}[LRPC codes]\label{def:LRPC}
Let $\Hv = (h_{ij})_{\substack{1\leqslant i \leqslant n-k\\ 1
    \leqslant j \leqslant n}} \in \Fqm^{(n-k)\times n}$ be a full-rank matrix such that its coordinates generate an $\Fq$-subspace $F$ of small dimension $d$:
\[
F = \langle h_{ij}\rangle_{\Fq}.
\]
Let $\C$ be the code with parity-check matrix $\Hv$. By definition, $\C$ is an $[n,k]_{q^m}$ LRPC code of dual weight $d$.
Such a matrix $\Hv$ is called a homogeneous matrix of weight $d$ and support $F$.
\end{definition}

We can now define ideal LRPC codes similarly to our definition of ideal codes. 

\begin{definition}[Ideal LRPC codes]\label{def:I-LRPC}
Let $F$ be an $\Fq$-subspace of dimension $d$ of $\Fqm$, let
$(\hv_1,\hv_2)$ be two vectors of $\Fqm^k$ of support $F$ and let $P
\in \Fq[X]$ be a polynomial of degree $k$.
Let
$$
\Hv_1 = \begin{pmatrix}
\hv_1 \\
X\hv_1 \mod P \\
\vdots \\
X^{k-1}\hv_1 \mod P
\end{pmatrix}^T
\text{ and }
\Hv_2 = \begin{pmatrix}
\hv_2 \\
X\hv_2 \mod P \\
\vdots \\
X^{k-1}\hv_2 \mod P
\end{pmatrix}^T.
$$

When the matrix $\Hv = (\Hv_1|\Hv_2)$ has rank $k$ over $\Fqm$,
the code $\C$ with parity check matrix $\Hv$ is called an ideal LRPC code of type $[2k,k]_{q^m}$.
\end{definition}

Since $P \in \Fq[X]$, the support of $X^i\hv_1$ is still $F$ for all $1 \leqslant i \leqslant k-1$. Hence the necessity to choose $P$ with coefficients in the base field $\Fq$ to keep the LRPC structure of the ideal code.

\subsection{A basic decoding algorithm}
\label{subsec:basic-decoding}

\begin{problem}[Decoding LRPC codes] \label{prob:decoding_lrpc} Given $\Hv = (h_{ij})_{\substack{1\leqslant i \leqslant n-k\\ 1
    \leqslant j \leqslant n}} \in \Fqm^{(n-k)\times n}$ a parity-check matrix of an LRPC code such that $h_{ij} \in F$ a subspace of $\Fqm$ of dimension $d$, a syndrome $\sv \in \Fqm^{n-k}$, and an integer $r$, the problem is to find a subspace $E$ of dimension at most $r$ such that there exists $\ev \in E^n$, $\Hv \ev^\intercal = \sv^\intercal$.
\end{problem}

Traditionally the decoding operation consists in finding not only the error support $E$ but also the exact vector $\ev$. However, in that case it is only a trivial algebraic computation to find the vector $\ev$ when $E$ is known, that is why we confuse both.

We denote by $EF$ the subspace generated by the product of the elements of $E$ and $F$:
\[
EF = \left\langle \{ ef, e \in E, f\in F\} \right\rangle
\]
In the typical case $\dim EF = rd$. For the considered parameters, it can happen that $\dim EF < rd$, but this case is also covered without modification.\\

A basic decoding algorithm is described in Algorithm~\ref{algo:RSR}. In the case where the syndrome $\sv$ is indeed generated by $\Hv\ev^\intercal$ where $\ev$ is in a support $E$, the coordinates of $\sv$ are in a product space $EF$ .The general idea of the algorithm is to use the fact that we know a parity-check matrix $\Hv$ of the LRPC code such that each of its coordinates $h_{ij}$ belongs to an $\Fq$-subspace $F$ of $\Fqm$ of small dimension $d$, hence the
subspace $S= \vect{s_1,\dots,s_{n-k}}$ generated by the coordinates of the syndrome 
enables one to recover the whole product space $EF$. The knowledge of both $EF$ and $F$ enables to recover $E$.
This approach is very similar to the classical decoding procedure
of BCH codes for instance, where one recovers the error-locator polynomial, which gives
the support of the error.\\

\begin{algorithm}[H]\label{algo:RSR}
\KwData{ $F = \langle f_1, ..., f_d\rangle$ an $\Fq$-subspace of $\Fqm$, $\sv=(s_1,\cdots,s_{n-k}) \in \Fqm^{(n-k)}$ a syndrome of an error $\ev$ of weight $r$ and of support $E$ }
\KwResult{A candidate for the vector space $E$}

\Comment{\textbf{Part 1: } Compute the vector space $\EF$}
Compute $S=\langle s_1,\cdots,s_{n-k}\rangle$

\Comment{\textbf{Part 2: } Recover the vector space $E$}
$E \leftarrow \bigcap_{i=1}^d f_i^{-1}S $
\Return{$E$}

\caption{Rank Support Recovery (\RSR) algorithm }
	
\end{algorithm}

\paragraph{Notation.}
For all $i$ we denote $S_i$ the space $f_i^{-1}S$.

\paragraph{Probability of failure.}

There are two cases for which the decoding algorithm can fail:
\begin{itemize}
	\item $S \subsetneq EF$, the syndrome coordinates do not generate the entire space $EF$, or
	\item $E \subsetneq S_1 \cap \dots \cap S_d$, the chain of intersections generates a space of larger dimension than $E$.
\end{itemize}

From~\cite{aragon2019low} we have that the probability of the first failure case $S \subsetneq EF$ is less than $q^{rd - (n-k) - 1}$.
In~\cite{ABDGHRTZABBBO19}, under the assumption that  the $S_i$ behave as random subspaces containing $E$ (which is validated by simulations), it is proven that the probability of the second failure case $E \subsetneq S_1 \cap \dots \cap S_d$ is less than $q^{-(d-1)(m-rd-r)}$. This leads to the following proposition from~\cite{ABDGHRTZABBBO19}:
\begin{proposition}
\label{prop:basic-DFR}
The Decoding Failure Rate of algorithm \ref{algo:RSR} is bounded from above by:
$$q^{-(d-1)(m-rd-r)} + q^{rd - (n-k) - 1}$$
\end{proposition}

\paragraph{Computational cost of decoding.}

According to \cite{aragon2019low}, the computational cost of the decoding algorithm is in $\mathcal O(4 r^2 d^2 m + n^2r)$ operations in the base field $\Fq$.\\

There is an improved version of this decoding algorithm which was presented in \cite{aragon2019low}. However, we do not need these improvements in the present document.

\subsection{LRPC codes indistinguishability}

LRPC codes are easy to hide since we only need to reveal their systematic parity-check matrix. Due to their weak algebraic structure, it is hard to distinguish an LRPC code in its systematic form and a random systematic matrix. We can now introduce formally this problem, on which LRPC cryptosystems, and thus ours, are based.

\begin{problem}[LRPC codes decisional problem - $\LRPC$]\label{prob:IndLRPC} Given a matrix $\Hv \in \Fqm^{(n-k)\times k}$, distinguish whether the code $\C$ with the parity-check matrix $(\Iv_{n-k}|\Hv)$ is a random  code or an LRPC code of weight $d$.
\end{problem}
	
The problem can also be stated as: distinguish whether $\Hv$ was sampled uniformly at random or as $\Av^{-1}\Bv$ where the matrices $\Av$ (of size $n-k\times n-k$) and $\Bv$ (of size $n-k\times k$) have the same support of small dimension $d$. The structured variant of the above problem follows immediately.

\begin{problem}[Ideal LRPC codes decisional problem - $\ILRPC$]\label{prob:IndILRPC} Given a polynomial $P \in \Fq[X]$ of degree $n$ and a vector $\hv \in \Fqm^n$, distinguish whether the ideal code $\C$ with the parity-check matrix generated by $\hv$ and $P$ is a random ideal code or an ideal LRPC code of weight $d$.
\end{problem}

Again, the problem can also be stated as: distinguish whether $\hv$ was sampled uniformly at random or as $\xv^{-1}\yv \mod P$ where the vectors $\xv$ and $\yv$ have the same support of small dimension $d$.\\

The hardness of these decisional problems is presented in Section \ref{sec:secu}.

%% file: multi-lrpc.tex
\section{LRPC with multiple syndromes} \label{sec:multi-LRPC}

\subsection{General idea}

The decoding algorithm presented in the previous section has a probability of failure whose main component is $q^{rd-(n-k)-1}$ (see Proposition \ref{prop:basic-DFR}) so it forces one to have a large $n$ in an LRPC-cryptosystem in order to obtain a DFR below $2^{-128}$. To overcome this constraint, we made the observation that when several syndromes with same error support $(\sv_1, ..., \sv_{\ell})$ were used in the decoding algorithm, the DFR was decreasing. This fact is the cornerstone of our new cryptosystem. We describe below the associated decoding problem.\\

\begin{problem}[Decoding LRPC codes with mutliple syndromes] \label{prob:decoding_lrpc_multi} Given $\Hv \in \Fqm^{(n-k)\times n}$ a parity-check matrix of an LRPC code of dimension d and support $F \subset \Fqm$, a set of $\ell$ syndromes $\sv_i \in \Fqm^{n-k}$ for $1 \leqslant i \leqslant \ell$, and an integer $r$, the problem is to find a subspace $E$ of dimension at most $r$ such that there exists an error matrix $\Vv \in E^{n \times \ell}$ satisfying $\Hv \Vv = \Sv$ where the i-th column of $\Sv$ is equal to $\sv_i^\intercal$.
\end{problem}

In order to solve this decoding problem, we introduce the Rank Support Recovery algorithm with multiple syndromes (Algorithm \ref{algo:multi-RSR}). It is exactly the same as Algorithm \ref{algo:RSR}, but several columns are given to compute the syndrome space $S$. Intuitively, because the syndrome matrix $\Hv\Vv$ has $(n-k) \times \ell$ coordinates inside the space $EF$ of dimension $rd$, we would expect the Decoding Failure Rate of this new algorithm to be approximately $q^{rd-(n-k)\ell}$. Actually, because the coordinates of $\Hv\Vv$ are not independent between each other, the result is not trivially established and requires technical lemmas which are presented in Section \ref{sec:proof-multi-RSR}.\\

\begin{algorithm}[H]\label{algo:multi-RSR}
	\KwData{ $F = \langle f_1, ..., f_d\rangle$ an $\Fq$-subspace of $\Fqm$, $\Sv=(s_{ij}) \in \Fqm^{(n-k)\times \ell}$ the $\ell$ syndromes of error vectors of weight $r$ and support $E$}
	\KwResult{A candidate for the vector space $E$}
	
	\Comment{\textbf{Part 1: } Compute the vector space $\EF$}
	Compute $S=\langle s_{11},\cdots,s_{(n-k)\ell}\rangle$
	
	\Comment{\textbf{Part 2: } Recover the vector space $E$}
	$E \leftarrow \bigcap_{i=1}^d f_i^{-1}S $
	
	\Return{$E$}
	
	\caption{Rank Support Recovery (\RSR) algorithm with multiple syndromes}
	
\end{algorithm}

\medskip

In the following subsection, we describe our new scheme and its ideal variant, then study the Decoding Failure Rate.

\subsection{Description of the scheme (LRPC-MS)}
\label{subsec:descriptionKEM}

\begin{definition}
A Key Encapsulation Mechanism KEM = $(\KeyGen,\Encap,\Decap)$ is a triple of probabilistic algorithms together with a key space $\mathcal{K}$. The key generation algorithm $\KeyGen$ generates a pair of public and secret keys $(\pk,\sk)$. The encapsulation algorithm $\Encap$ uses the public key $\pk$ to produce an encapsulation $c$ and a key $K \in \mathcal{K}$. Finally $\Decap$, using the secret key $\sk$ and an encapsulation $c$, recovers the key $K \in \mathcal{K}$, or fails and returns $\bot$.
\end{definition}

Our scheme contains a hash function $G$ modeled as a random oracle.
\begin{itemize}
\item $\KeyGen(1^\seck)$: 
\begin{itemize}
\item choose uniformly at random a subspace $F$ of $\Fqm$ of dimension $d$ and sample an couple of homogeneous matrices of same support $\Uv = (\Av | \Bv) \rand F^{(n-k)\times (n-k)} \times F^{(n-k)\times k}$ such that $\Av$ is invertible.
\item compute $\Hv = (\Iv_{n-k} | \Av^{-1}\Bv)$.
\item define $\pk = \Hv$ and $\sk = (F, \Av)$.
\end{itemize}
\item $\Encap(\pk)$:
\begin{itemize}
\item choose uniformly at random a subspace $E$ of $\Fqm$ of dimension $r$ and sample a matrix $\Vv \rand E^{n \times\ell}$.
\item compute $\Cv = \Hv\Vv$.
\item define $K = G(E)$ and return $\Cv$.
\end{itemize}
\item $\Decap(\sk)$:
\begin{itemize}
\item compute $\Sv = \Av\Cv~(= \Uv\Vv)$
\item recover $E \leftarrow \RSR(F, \Sv, r)$ (Algorithm \ref{algo:multi-RSR}).
\item return $K = G(E)$ or $\bot$ (if $\RSR$ failed).
\end{itemize}
\end{itemize}

We need to have a common representation of a subspace of dimension $r$ of $\Fqm$. The natural way is to choose the unique matrix $\Mv \in \Fq^{r\times m}$ of size $r\times m$ in its reduced row echelon form such that the rows of $\Mv$ are a basis of $E$.

An informal description of this scheme can be found in Figure~\ref{fig:new_proto}. We deal with the semantic security of the KEM in Section~\ref{sec:secu}.

\begin{figure*}[bt]

\resizebox{\textwidth}{!}{
\fbox{\hspace{-.0cm}%
\begin{minipage}{1.15\textwidth}
\begin{minipage}{.4225\textwidth}
	\parbox{\textwidth}{\centering
		\underline{Alice}}
\end{minipage} 
\begin{minipage}{.1125\textwidth}
	\parbox{\textwidth}{}
\end{minipage} ~
\begin{minipage}{.4225\textwidth}
	\parbox{\textwidth}{\centering
		\underline{Bob}}
\end{minipage} \\[2mm]
\begin{minipage}{.4225\textwidth}
	\parbox{\textwidth}{\centering%
		choose $F$ of dimension $d$ at random\\
		$\Uv = (\Av | \Bv) \rand F^{(n-k)\times n}$, $\Hv = (\Iv_{(n-k)} | \Av^{-1}\Bv)$ syst. form of $\Uv$\\[10mm]
		$\Sv = \Av\Cv$\\
		$E \leftarrow \RSR(F, \Sv, r)$\\
		~\\[2mm]
		\fbox{$G\left(E\right)$}%
	}%
\end{minipage} ~
\begin{minipage}{.1125\textwidth}
	\parbox{\textwidth}{\centering%
		~\\[7mm]
		$\xrightarrow{~~~~\mathbf{H}~~~~}$\\[6mm]
		$\xleftarrow{~~~~\mathbf{C}~~~~}$\\[8mm]
		
		\textsc{Shared Secret}%
	}
\end{minipage} ~
\begin{minipage}{.4225\textwidth}
	\parbox{\textwidth}{\centering%
		~\\[8mm]
		choose $E$ of dimension $r$ at random\\
		$\Vv \rand E^{n\times\ell}$ \\
		$\Cv = \Hv \Vv$\\[12mm]
		\fbox{$G\left(E\right)$}%
	}
\end{minipage}
\end{minipage}
}%
}
\caption{\label{fig:new_proto}Informal description of our new Key Encapsulation Mechanism LRPC-MS. $\Hv$ constitutes the public key.}
\end{figure*}

\subsection{Description of the scheme with ideal structure (ILRPC-MS)}
\label{subsec:descriptionIdealKEM}

Our scheme contains a hash function $G$ modeled as a random oracle and an irreducible polynomial $P \in \Fq[X]$ of degree $k$, which are public parameters.
\begin{itemize}
\item $\KeyGen(1^\seck)$: 
\begin{itemize}
\item choose uniformly at random a subspace $F$ of $\Fqm$ of dimension $d$ and sample a couple of 	vectors $(\xv,\yv) \rand F^k \times F^k$ such that $\Supp(\xv)=\Supp(\yv)=F$.
\item compute $\hv = \xv^{-1}\yv \mod P$.
\item define $\pk = (\hv, P)$ and $\sk = (\xv,\yv)$.
\end{itemize}
\item $\Encap(\pk)$:
\begin{itemize}
\item choose uniformly at random a subspace $E$ of $\Fqm$ of dimension $r$ and sample a tuple of 	vectors $(\ev_1,...,\ev_{2\ell}) \rand (E^k)^{2\ell}$ such that $\Supp(\ev_1,...,\ev_{2\ell})=E$.
\item compute for all $1 \leq i \leq \ell$, $\cv_i = \ev_{2i-1} + \ev_{2i}\hv \mod P$.
\item define $K = G(E)$ and return $(\cv_1,...,\cv_{\ell})$.
\end{itemize}
\item $\Decap(\sk)$:
\begin{itemize}
\item compute for all $i$, $\xv\cv_i = \xv\ev_{2i-1}+ \yv\ev_{2i} \mod P$ and 
\item compute $\Sv = (\xv\cv_1,...,\xv\cv_{\ell})$
\item recover $E \leftarrow \RSR(F, \Sv, r)$ (Algorithm \ref{algo:multi-RSR}).
\item return $K = G(E)$ or $\bot$ (if $\RSR$ failed).
\end{itemize}
\end{itemize}

An informal description of this scheme is found in Figure~\ref{fig:new_proto_ideal}. As for the non-ideal scheme, we deal with the semantic security of the KEM in Section~\ref{sec:secu}.

\begin{figure*}[bt]
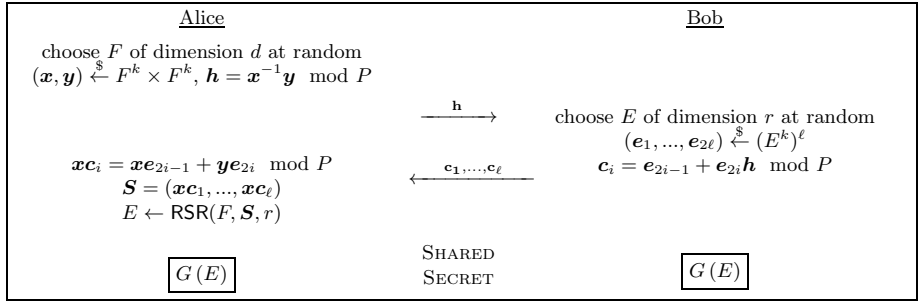


\resizebox{\textwidth}{!}{
\fbox{\hspace{-.0cm}%
\begin{minipage}{1.15\textwidth}
	\begin{minipage}{.4225\textwidth}
		\parbox{\textwidth}{\centering
			\underline{Alice}}
	\end{minipage} 
	\begin{minipage}{.1125\textwidth}
		\parbox{\textwidth}{}
	\end{minipage} ~
	\begin{minipage}{.4225\textwidth}
		\parbox{\textwidth}{\centering
			\underline{Bob}}
	\end{minipage} \\[2mm]
	\begin{minipage}{.4225\textwidth}
		\parbox{\textwidth}{\centering%
			choose $F$ of dimension $d$ at random\\
			$(\xv,\yv) \rand F^k\times F^k$, $\hv = \xv^{-1}\yv \mod P$\\[10mm]
			$\xv\cv_i = \xv\ev_{2i-1} + \yv\ev_{2i} \mod P$\\
			$\Sv = (\xv\cv_1,...,\xv\cv_{\ell})$\\
			$E \leftarrow \RSR(F, \Sv, r)$\\
			~\\[2mm]
			\fbox{$G\left(E\right)$}%
		}%
	\end{minipage} ~
	\begin{minipage}{.1125\textwidth}
		\parbox{\textwidth}{\centering%
			~\\[7mm]
			$\xrightarrow{~~~~\mathbf{h}~~~~}$\\[6mm]
			$\xleftarrow{~~~~\mathbf{c_1,...,c_{\ell}}~~~~}$\\[8mm]
			
			\textsc{Shared Secret}%
		}
	\end{minipage} ~
	\begin{minipage}{.4225\textwidth}
		\parbox{\textwidth}{\centering%
			~\\[8mm]
			choose $E$ of dimension $r$ at random\\
			$(\ev_1,...,\ev_{2\ell}) \rand (E^k)^{\ell}$ \\
			$\cv_i = \ev_{2i-1} + \ev_{2i}\hv \mod P$\\[12mm]
			\fbox{$G\left(E\right)$}%
		}
	\end{minipage}
\end{minipage}
}%
}
\caption{\label{fig:new_proto_ideal}Informal description of our new Key Encapsulation Mechanism with ideal structure ILRPC-MS. $\mathbf{h}$ constitutes the public key.}
\end{figure*}

\subsection{Decoding Failure Rate of our scheme}
\label{subsec:DFR}

The Decoding Failure Rate (DFR) of our scheme is the probability of failure of the Rank Support Recovery algorithm with multiple syndromes described in Algorithm \ref{algo:multi-RSR}.
As stated in Section \ref{subsec:basic-decoding}, the two cases that can provoke a failure of the algorithm are:

\begin{itemize}
	\item $S \subsetneq EF$, the coordinates of the matrix $\Uv\Vv$ do not generate the entire space $EF$, or
	\item $E \subsetneq S_1 \cap \dots \cap S_d$, the chain of intersections generate a space of larger dimension than $E$.
\end{itemize}

To study the probability of each case, we restrict ourselves to the case $\dim(EF) = rd$. Indeed, when  $\dim(EF) < rd$, the correctness of the algorithm is preserved, and the probabilities associated to the two sources of decoding failures are lower than in the case $\dim(EF) = rd$, since all the vector spaces will be of smaller dimensions. Hence this restriction will lead to an upper bound on the failure probability.\\

The first case of failure can be dealt with the following theorem, which will be fully proven in Section \ref{sec:proof-multi-RSR}. Its immediate corollary yields the probability of failure for the first case. We will assume for the rest of this document that $q = 2$ since the theorem is only proven in that case.

\begin{theorem}
	For $n_1 + n_2 \leq n$ and for $\Uv$ and $\Vv$ random variables chosen uniformly in $F^{n_1\times n}$ and $E^{n\times n_2}$ respectively, $\mathbb P(\Supp{(\Uv \Vv)} \neq EF) \leq n_1q^{rd-n_1n_2}$
\end{theorem}

\begin{corollary}
For $k \geq \ell$ and for $\Uv$ and $\Vv$ random variables chosen uniformly in $F^{(n-k)\times n}$ and $E^{n\times\ell}$ respectively, the probability that the syndrome space $S$ computed by the algorithm $\RSR(F, \Uv\Vv, r)$ is not equal to $EF$ is bounded by above by $(n-k)q^{rd-(n-k)\ell}$
\end{corollary}

As for the second failure case, $E \subsetneq S_1 \cap \dots \cap S_d$, we apply again the upper-bound $q^{-(d-1)(m-rd-r)}$, used in Section~\ref{sec:LRPC} for Proposition~\ref{prop:basic-DFR}. This leads to the following proposition:

\begin{proposition}
\label{prop:DFR}
For $k \geq \ell$ and for $\Uv$ and $\Vv$ random variables chosen uniformly in $F^{(n-k)\times n}$ and $E^{n\times\ell}$ respectively, the Decoding Failure Rate of algorithm \ref{algo:multi-RSR} $\RSR(F, \Uv\Vv, r)$ is bounded from above by:
$$q^{-(d-1)(m-rd-r)} + (n-k)q^{rd-(n-k)\ell}$$
\end{proposition}

This proposition extends immediately to the ideal case without modifications.\\

\subsection{Impact on the asymptotic range of parameters}

By reducing the decoding failure rate, the multiple syndrome approach fundamentally changes the zone of parameters that we consider for our cryptosystem.\\

In previous LRPC code-based cryptosystems, the decoding failure rate imposed the choice of $r$ and $d$ to be below $\sqrt n$ because of the need for $rd < n-k$ (cf. Proposition \ref{prop:basic-DFR}).\\

In this work, we can choose $r$ and $d$ bigger than $\sqrt n$. We will show that it is even possible to reach the rank Gilbert-Varshamov bound $d_{RGV}$. To simplify the rest of the analysis we will consider half-rate codes only, for which $k = n / 2$. In that case there exists a simple upper bound for $d_{RGV}$:

\begin{lemma}
For $k = n/2$, $d_{RGV}(m,n) \leq n/2$.
\end{lemma}

\begin{proof}
When apply the asymptotic formula of $d_{RGV}$ (\cite{aragon2019durandal}, §2.4) to the case $k = n/2$, we get
$$d_{RGV}(m,n) = \frac{m+n-\sqrt{m^2+n^2}}{2}.$$
\noindent We then calculate
$$\frac{\partial d_{RGV}}{\partial m} = \frac{1}{2} \left( 1 - \frac{m}{\sqrt{m^2 + n^2}}\right)$$
which is a positive quantity so we have
$$
d_{RGV}(m,n) \leq \lim_{m \rightarrow \infty} d_{RGV}(m,n) = n/2.
$$
 \qed
\end{proof}

As a result, if we choose $r$ and $d$ to be egal to $d_{RGV}$, we get an asymptotic condition on the number of syndromes $\ell$ due to Proposition \ref{prop:DFR}. Indeed $(n-k)\ell = rd + o(1)$, which gives: $n\ell / 2 = d_{RGV}^2 + o(1)$. We then deduce an asymptotic upper bound on $\ell$ to be $r = d_{RGV}$. To the best of our knowledge, the range where $\ell \leq r$ is a hard parameter range for which the $\RSL$ problem has no known polynomial attacks.\\

The fact that we can choose $r$ and $d$ on the rank Gilbert-Varshamov bound has two major implications:
\begin{itemize}
	\item Algebraic attacks against the $\RSD$ problem are more difficult when $r$ gets closer to $d_{RGV}$.
	\item The secret parity check matrix $\Uv$ is an homogeneous of weight $d_{RGV}$ so the minimal distance of the dual of the resulting LRPC code is about $d_{RGV}$, just like a random code. It gives more confidence in the indistinguishably of the public matrix $\Hv$ ($\LRPC$ problem).
\end{itemize}

Our proposal is the only code-based cryptosystem with structural masking that has such an interesting property for the distinguishing problem.

\subsection{Reducing the value of parameter $m$ (LRPC-xMS)}
\label{subsec:reducing-m}

In this subsection we introduce a variation that further reduces the DFR at the cost of additional computations.

When choosing the parameters for our cryptosystem, we generally want to take $m$ as small as possible in order to reduce the sizes of both keys and ciphertexts. However, reducing $m$ might lead to decoding failures where $E \subsetneq S_1 \cap \dots \cap S_d$, meaning that $E' = S_1 \cap \dots \cap S_d$ has a dimension strictly greater than $r$. In this section, we show that with overwhelming probability $E'$ is of dimension at most $r+1$. If the dimension is $r+1$, it is possible to enumerate all possible subspaces of dimension $r$ of $E'$ and check for each of these subspaces if it is the target space $E$ or not. We modify slightly the cryptosystem so that the receiver can efficiently check whether a given subspace is $E$: we use a hash function (different from the one used to compute the shared secret) to add a hash of $E$ to the ciphertext. By hashing each of the enumerated subspaces it is thus possible to find which one is $E$.

The extended Rank Support Recovery Algorithm with multiple syndromes that allows to recover $E$ even when $E'$ is of dimension $r+1$, is presented in Algorithm \ref{algo:multi-RSR-hash}. The improved Key Encapsulation Mechanism LRPC-xMS taking advantage of this algorithm is informally presented in Figure \ref{fig:new_proto_improved}. Take note that Algorithm $\xRSR$ needs a hash as a parameter, therefore the hashed value $G'(E)$ has to be sent in the ciphertext. The second hash function $G'$ is a public parameter of the system.\\

\begin{algorithm}[h]\label{algo:multi-RSR-hash}
	\KwData{ $F = \langle f_1, ..., f_d\rangle$ an $\Fq$-subspace of $\Fqm$, $\Sv=(s_{ij}) \in \Fqm^{(n-k)\times \ell}$ the $\ell$ syndromes of error vectors of weight $r$ and support $E$ and a hash $\mathcal H$}
	\KwResult{A candidate for the vector space $E$}
	
	\Comment{\textbf{Part 1: } Compute the vector space $\EF$}
	Compute $S=\langle s_{11},\cdots,s_{(n-k)\ell}\rangle$
	
	\Comment{\textbf{Part 2: } Recover the vector space $E'$}
	$E' \leftarrow \bigcap_{i=1}^d f_i^{-1}S $
	
	\Comment{\textbf{Part 3: } Check candidate subspaces for dim(E') = r+1}
	\uIf{$\dim E' = r$}{\Return{$E'$}}
	\uElseIf{$\dim E' > r+1$}{\Return{$\bot$}}
	\Else{
	\For{every $E_C \subset E'$, $\dim E_C = r$}{
		\If{$G'(E_C) == \mathcal H$}{\Return{$E_C$}}
	}
	\Return{$\bot$}}
	
	\caption{Extended Rank Support Recovery (\xRSR) algorithm with multiple syndromes}
	
\end{algorithm}

\begin{figure*}[bt]
	
\resizebox{\textwidth}{!}{
\fbox{\hspace{-.0cm}%
\begin{minipage}{1.15\textwidth}
\begin{minipage}{.4225\textwidth}
	\parbox{\textwidth}{\centering
		\underline{Alice}}
\end{minipage} 
\begin{minipage}{.1125\textwidth}
	\parbox{\textwidth}{}
\end{minipage} ~
\begin{minipage}{.4225\textwidth}
	\parbox{\textwidth}{\centering
		\underline{Bob}}
\end{minipage} \\[2mm]
\begin{minipage}{.4225\textwidth}
	\parbox{\textwidth}{\centering%
		choose $F$ of dimension $d$ at random\\
		$\Uv = (\Av | \Bv) \rand F^{(n-k)\times n}$, $\Hv = (\Iv_{n-k} | \Av^{-1}\Bv)$ syst. form of $\Uv$\\[10mm]
		$\Sv = \Av\Cv$\\
		$E \leftarrow \xRSR(F, \Sv, r, G'(E))$\\
		~\\[2mm]
		\fbox{$G\left(E\right)$}%
	}%
\end{minipage} ~
\begin{minipage}{.1125\textwidth}
	\parbox{\textwidth}{\centering%
		~\\[7mm]
		$\xrightarrow{~~~~\mathbf{H}~~~~}$\\[6mm]
		$\xleftarrow{~~~~\mathbf{C}, G'(E)~~~~}$\\[12mm]
		
		\textsc{Shared Secret}%
	}
\end{minipage} ~
\begin{minipage}{.4225\textwidth}
	\parbox{\textwidth}{\centering%
		~\\[8mm]
		choose $E$ of dimension $r$ at random\\
		$\Uv \rand E^{n\times \ell}$ \\
		$\Cv = \Hv \Vv$\\[18mm]
		\fbox{$G\left(E\right)$}%
	}
\end{minipage}
\end{minipage}
}%
}
\caption{\label{fig:new_proto_improved}Informal description of our improved new Key Encapsulation Mechanism LRPC-xMS. $\mathbf{H}$ constitutes the public key.}
\end{figure*}

This algorithm fails only when $\dim(S_1 \cap \dots \cap S_d) > r+1$. In order to provide an upper bound for this probability let us first start with a quick lemma to bound Gaussian binomial coefficients: 

\begin{lemma}
\label{lem:gauss}
For all $c \le m$ we have $\cg{c}{m} \leq \frac{1}{\phi(q^{-1})} \, q^{c(m-c)}$
where $\phi$ is the Euler function given by

$$ \phi(x) = \prod_{k=1}^{\infty} (1 - x^k) \textnormal{ for } |x| < 1.$$
\end{lemma}

\begin{proof}
We first note that:	
\begin{align*}
\prod_{i=0}^{c-1} (q^c-q^i) &= q^{c^2} \prod_{i=0}^{c-1} (1 - q^{i-c})\\
&=q^{c^2} \prod_{i=1}^{c} (1 - q^{-i})\\
&\geq q^{c^2} \prod_{i=1}^{\infty} (1 - q^{-i})\\
&\geq q^{c^2} \phi(q^{-1})
\end{align*}
We also have immediately:
$$
\prod_{i=0}^{c-1} (q^m-q^i) \leq q^{cm}
$$
As a result,
\begin{align*}
\cg{c}{m} &= \prod_{i=0}^{c-1} \frac{q^m-q^i}{q^c-q^i}\\
&\leq \frac{q^{cm}}{q^{c^2} \phi(q^{-1})}\\
&\leq \frac{1}{\phi(q^{-1})} \, q^{c(m-c)} \tag*{\qed}
\end{align*}
\end{proof}

Next, using the same assumption as for Propositions~\ref{prop:basic-DFR} and~\ref{prop:DFR} on the independence of the spaces $S_i$ with respect to the dimension of intersections, we prove the following general result for small values of $c$:

\begin{proposition} \label{prop:prob_E_larger}
  Let $c \geq 1$, $S = \EF$, $\{f_1, \dots, f_d\}$ a basis of $F$, and $S_i = f_i^{-1}S$. Then $\frac{1}{\phi(q^{-1})} \, q^{c(rd-r-c + (d-1)(rd-m))}$ is an upper bound for $\Prob(\dim(\bigcap\limits_{i=1}^d S_i) \geq r + c)$.
\end{proposition}

\begin{proof}
Since $S \subset EF$ we have $E \subset S_i$ for $i = 1, \dots, d$, hence each quotient vector space $S_i / E$ is a subspace of $\Fqm / E$ and we have the following equivalence:

\begin{align*}
\dim(\bigcap\limits_{i=1}^d S_i) \geq r+c &\Leftrightarrow \exists A, \dim A = c, (\bigcap\limits_{i=1}^d S_i / E) \supset A\\
&\Leftrightarrow \exists A, \dim A = c, \text{for all } i=1..d, S_i / E \supset A
\end{align*}

%
As every $S_i / E$ behaves as being an independent subspace of $\Fqm / E$ of dimension $rd-r$ with respect to the intersection, the probability that a random $A$ of dimension $c$ is in $S_i / E$ is $$\frac{\cg{c}{rd-r}}{\cg{c}{m-r}} = \prod_{i=0}^{c-1} \frac{q^{rd-r}-q^i}{q^{m-r}-q^i} \approx q^{c(rd-m)}$$ because $c \ll rd-r < m$. By enumerating over every subspace of dimension $c$ of (say) $S_1 / E$, we obtain:

\begin{align*}
	\Prob(\dim(\bigcap\limits_{i=1}^d S_i / E) \geq c) & \leq \sum\limits_{A \subset S_1/E, \dim(A) = c} \,\prod\limits_{j = 2, \dots, d} \Prob(A \subset S_j / E)\\
	& \leq \sum\limits_{A \subset S_1/E, \dim(A) = c} q^{(d-1) \times c(rd-m)}\\
	& \leq \cg{c}{rd-r} q^{(d-1) \times c(rd-m)}\\
	& \leq \frac{1}{\phi(q^{-1})} \, q^{c(rd-r-c)} q^{(d-1) \times c(rd-m)} \qquad \textnormal{(Lemma \ref{lem:gauss})} \tag*{\qed}
\end{align*}
\end{proof}

\begin{remark} \label{remark:prob_E_plus_1}
	The upper bound from proposition \ref{prop:prob_E_larger} in the case $c=1$ can be refined by taking the cardinal of $S_1 / E$ as the number of subspaces of dimension $1$ of $S_1 / E$ instead of approximating the Gaussian binomial. This leads to:
	
	$$\Prob(\dim(\bigcap\limits_{i=1}^d S_i) \geq r + 1) \leq q^{rd-r + (d-1)(rd-m)}$$
	
	Which corresponds to the formula given subsection \ref{subsec:basic-decoding}.
\end{remark}

We can now present the following estimation on the decoding failure rate of our improved cryptosystem:

\begin{proposition}
\label{prop:DFR-extended}
For $k \geq \ell$ and for $\Uv$ and $\Vv$ random variables chosen uniformly in $F^{(n-k)\times n}$ and $E^{n\times\ell}$ respectively:
$$\frac{1}{\phi(q^{-1})} \, q^{2(rd-r-2+(d-1)(rd-m))} + (n-k)2^{rd-(n-k)\ell}$$
is an upper bound on the Decoding Failure Rate of
algorithm \ref{algo:multi-RSR-hash} $\xRSR(F, \Uv\Vv, r)$.
\end{proposition}

\begin{remark}
We could have dealt with larger dimensions, i.e. $\dim(\bigcap\limits_{i=1}^d S_i) = r + c$ with $c \geq 2$. The idea would be to use the same technique used in \cite{AGHT18} for solving the $\RSD$ problem, except that we already have the knowledge of a vector space $E'$ of dimension $r+c$ that contains the error support $E$. The idea is as follows: let $\{e'_1, \dots, e'_{r+c}\}$ be a basis of $E'$. Rewriting the system $\Hv \ev^\intercal = \sv$ over $\Fq$ with the knowledge that $\ev \in E'^n$ yields a system with $n(r+c)$ unknowns and $(n-k)rd$ equations. By solving this linear system, we can recover the error vector $\ev$ and thus its support $E$. However it would be difficult to evaluate the true randomness of the system $\Hv \ev^\intercal = \sv$ in such special cases, that is why we decided not to consider such cases.
\end{remark}

%% file: proof_multi_rsr.tex
\section{Dimension of the support of the product of homogeneous matrices} \label{sec:proof-multi-RSR}

In this section we prove the following theorem, which is required to prove the correctness of the multi-syndrome approach presented in the previous section.
We fix $E$ and $F$ subspaces of $\Fqm$ of dimension $r$ and $d$ respectively such that $EF$ is of dimension $rd$. Remember that we have $q = 2$.

\begin{theorem}\label{thm:support-product}
	For $n_1+n_2\leq n$ and for $\Uv$ and $\Vv$ random variables chosen uniformly in $F^{n_1\times n}$ and $E^{n\times n_2}$ (respectively), $\mathbb P(\Supp(\Uv \Vv) \neq EF) \leq n_1 q^{rd-n_1n_2}$
\end{theorem}

A first idea which may come to mind when trying to prove this theorem would be to use the Leftover Hash Lemma \cite{impagliazzo1989pseudo} (LHL) in order to prove that the statistical distribution of $\Uv\Vv$ is $\varepsilon$-close to the uniform statistical distribution on $EF^{n_1\times n_2}$.
However, the total number of different couples $(\Uv, \Vv)$ is equal to $\dim F^{n_1n} \dim E^{n_2n} = rd^n r^{n_2} d^{n_1}$ and the number of matrices in $EF^{n_1\times n_2}$ is $rd^{n_1n_2}$. In a usual code-based cryptography setting where $n_1 \approx n_2 \approx n/2$ and $r \approx d$, we get that $rd^n r^{n_2} d^{n_1} \ll rd^{n_1n_2}$ therefore we cannot expect to use the LHL.

At first sight, this is quite an issue, as proving the statement of our theorem without standard statistical arguments can be quite complex, or impossible. The rest of the section presents a five stage proof of the theorem (main body and 4 lemmas), using algebraic arguments. Our approach is to study the distribution of $\phi(\Uv\Vv)$ for a linear form $\phi$ on $EF$. We show that the distribution of $\phi(\Uv\Vv)$ is uniform in a subspace of $\F^{n_1\times n_2}$ whose dimension is depending on the rank of $\phi$ viewed as a tensor in $E\otimes F$ and on a simple condition on matrix $\Uv$.

\subsection{Preliminary results on binary matrices}

%

%

\begin{lemma}\label{lem:majrank}
	For a uniformly random binary matrix $\Mv$ of size $m \times n$ with $m \leq n$ and for $0 < i \leq m$, $\mathbb P(\rank(\Mv) \leq m - i) \leq 2^{i(m-n)}$.
\end{lemma}

\begin{proof}
Let $S$ be a subspace of $\{0,1\}^m$ of dimension $m-i$. The number of such possible subspaces is $\binom{m}{i}_2 \leq 2^{im}$.

For a uniformly random binary $m\times n$ matrix $\Mv$, since the $n$ columns of $\Mv$ are independent, $\mathbb P(\Supp(\Mv) \subset S) = 2^{-in}$. Then:
\begin{align*}
	\mathbb P(\Rank(\Mv) \leq m - i) &= \mathbb P(\bigcup_{S} \Supp(\Mv) \subset S)\\
	&\leq \sum_{S} \mathbb P(\Supp(\Mv) \subset S)\\
	&\leq 2^{i(m-n)} \tag*{\qed}
\end{align*}
\end{proof}

\begin{definition}
	For $s > 0$, let $R_s$ be the random variable defined as the rank of a uniformly random binary matrix of size $n_1\times ns$.
\end{definition}

\begin{lemma}\label{lem:exp-r_1}
	For $n_2 > 0$, $\Expec(2^{-n_2R_1}) \leq n_1 2^{-n_1n_2}$.
\end{lemma}

\begin{proof}
\begin{align*}
	\Expec(2^{-n_2R_1}) &= \sum_{i = 0}^{n_1} 2^{-n_2i} \mathbb P(R_1 = i)\\
	&= 2^{-n_1n_2}\Prob(R_1 = n_1) + \sum_{i = 0}^{n_1 - 1} 2^{-n_2i} \mathbb P(R_1 = i)\\
	&\leq 2^{-n_1n_2} + \sum_{i = 1}^{n_1} 2^{-n_2(n_1 -i)} \mathbb P(R_1 = n_1 - i)\\
	&\leq 2^{-n_1n_2} + \sum_{i = 1}^{n_1} 2^{-n_2(n_1 -i)} 2^{i(n_1 - n)} \qquad (\text{Lemma \ref{lem:majrank}})\\
	&\leq 2^{-n_1n_2} + \sum_{i = 1}^{n_1} 2^{i(n_2+n_1-n)-n_1n_2}\\
	&\leq 2^{-n_1n_2} + \sum_{i = 1}^{n_1} 2^{-n_1n_2} \qquad (n \geq n_1 + n_2)\\
	&\leq n_1 2^{-n_1n_2} \tag*{\qed}
\end{align*}
\end{proof}
Since $R_1 \stackrel{\mathcal L}{\leq} R_s$, we get an immediate corollary.
\begin{corollary}\label{cor:exp-r_s}
For $n_2 > 0$ and for $s > 0$, $\Expec(2^{-n_2R_s}) \leq n_1 2^{-n_1n_2}$.
\end{corollary}

\subsection{Proof of Theorem \ref{thm:support-product}}

We first fix $\phi$ a non-zero linear form from $EF$ to $\Fq$ and we will study the probabibilty that $\Supp(\Uv \Vv) \subset \ker(\phi)$. For a vector $\xv = (x_1,...,x_i) \in (EF)^i$, we will note $\phi(\xv)$ the
vector $(\phi(x_1),...,\phi(x_i))$. We use the similar abuse of notation for $\phi(\Xv)$ when $\Xv$ is a matrix.

Let $\phi_{b}$ be the non-zero bilinear form
\begin{eqnarray*}
	\phi_{b}: E\times F & \to  & \F\\
	(e,f)  & \mapsto & \phi(ef).
\end{eqnarray*}
Let $s = \Rank(\phi_{b})$ be the rank of this bilinear form. Then there exists a basis $(e_1,\ldots, e_r)$ of $E$ and a basis $(f_1,\ldots, f_d)$ of $F$ in which the matrix representation of $\phi_{b}$ is
$$\begin{pmatrix}
	\Iv_s & 0\\
	0 & 0
\end{pmatrix}$$
In the product basis of $EF$ $$(e_1,\ldots, e_r) \otimes (f_1,\ldots, f_d) = (e_1f_1, ..., e_1f_d, e_2f_1,..., e_rf_1,..., e_rf_d)$$ the expression of $\phi$ is very simple. For $x = \sum_{\substack{1\leq i\leq n\\1\leq j\leq n}} x_{ij}e_if_j$ we have
$$\phi( x ) = \sum_{1\leq i\leq s} x_{ii}.$$

%
Let $\uv=(u_1,\ldots, u_n)$ be a vector of $F^n$ and consider the map
\begin{eqnarray*}
	E^n & \to & \F\\
	\vv=(v_1,\ldots v_n)^\intercal   & \mapsto &\phi(\uv\vv)=\phi(
	u_1v_1+\cdots +u_nv_n).
\end{eqnarray*}

For $i=1\ldots n$, write $u_i=\sum_{j=1}^du_{ij}f_j$ the decomposition of $u_i$
along the basis of $F$ $(f_1,\ldots, f_d)$. Similarly write
$v_i=\sum_{j=1}^rv_{ij}e_j$ the decomposition of $v_i$ along the basis of $E$ $(e_1,\ldots, e_r)$. We clearly have:
\begin{equation}
	\label{eq:phi(uv)}
	\phi(\uv\vv) = \sum_{\substack{1\leq i\leq n\\1\leq j\leq s}}u_{ij}v_{ij}.
\end{equation}

Now let $\Uv$ be an $n_1\times n$ matrix of elements in $F$. 
Define $\Uv^s$ to be the $n_1\times sn$ {\em binary} matrix obtained from $\Uv$ by
replacing every one of its rows $\uv$ by its expansion
\[
u_{11},\ldots ,u_{1s},u_{21},\ldots ,u_{2s},\ldots u_{n1},\ldots ,u_{ns}
\]
as defined in \eqref{eq:phi(uv)}. It follows that we have:

\begin{lemma}
	Let $s = \rank(\phi_b)$, $\Uv$ be an $n_1\times n$ matrix of elements in $F$ and let $\varphi_{\Uv}$ be the
	map
	\begin{eqnarray*}
		\varphi_{\Uv}: E^n & \to &\F^{n_1}\\
		\vv & \mapsto &\phi(\Uv\vv).
	\end{eqnarray*}
	The rank of the map $\varphi_{\Uv}$ is equal to the rank of the $n_1\times sn$ binary
	matrix $\Uv^s$.
\end{lemma}

\begin{corollary}
\label{cor:law-rk-phi}
For $\Uv$ a random variable chosen uniformly in $F^{n_1\times n}$, $\Rank(\varphi_{\Uv}) \stackrel{\mathcal L}{=} R_s$ where $s$ is the rank of $\phi_b$.
\end{corollary}

%
Now that we know the probability distribution of the rank of $\varphi_{\Uv}$, we will give a probability on $\Supp(\Uv\Vv)$ depending on this rank.

\begin{lemma}
\label{lemma:support-UV}
Let $\Uv$ such that the above-defined $\varphi_{\Uv}$ is of rank $0 \leq i \leq n_1$.
Then for $\Vv$ a random variable chosen uniformly in $E^{n\times n_2}$, $\mathbb P(\Supp(\Uv \Vv) \subset \ker(\phi))  \leq q^{-i n_2}$
\end{lemma}

\begin{proof}
	Let $H = Im(\varphi_{\Uv})$
	Let $\Vv = (\vv_1,...,\vv_{n_2})$ the columns of $\Vv$.\\
	$\varphi_{\Uv}$ is a surjective homomorphism of finite abelian groups $E^n$ and $H$, so according to Theorem 8.5 in \cite{shoup2009computational}, for all $i$, $\Uv \vv_i$ is uniformly distributed. Thus because the columns of $\Vv$ are independent, $\phi(\Uv \Vv)$ is uniformly distributed in $H^{n_2}$.\\
	As a result, because $\Supp{(\Uv \Vv)} \subset \ker(\phi)$ if and only if $\phi(\Uv \Vv) = \mathbf 0$, $\mathbb P(\Supp{(\Uv \Vv)} \subset \ker(\phi)) \leq 1/|H^{n_2}| = q^{-i n_2}$.
\qed
\end{proof}






\begin{lemma}\label{lem:maj-kerphi}
	For a non-null linear form $\phi$ of $EF$, $\mathbb P(\Supp(\Uv \Vv) \subset \ker(\phi)) \leq \Expec(2^{-n_2R_1})$
\end{lemma}

\begin{proof}
Let $s > 0$ be the rank of $\phi_b$.
	\begin{align*}
		\mathbb P(\Supp(\Uv \Vv) \subset \ker(\phi)) ={}& \sum_{i = 0}^{n_1} \mathbb P(\Supp(\Uv \Vv) \subset \ker(\phi) | \Rank(\varphi_{\Uv}) = i)\, \mathbb P(\Rank(\varphi_{\Uv})  = i)\\
		\leq {}&  \sum_{i = 0}^{n_1} 2^{-in_2} \mathbb P(\Rank(\varphi_{\Uv})  = i) \qquad (\textrm{Lemma \ref{lemma:support-UV}})\\
		\leq {}& \Expec(2^{-n_2\,\Rank(\varphi_{\Uv})})\\
		\leq {}& \Expec(2^{-n_2R_s}) \qquad (\textrm{Corollary \ref{cor:law-rk-phi}})\\
		\leq {}& \Expec(2^{-n_2R_1}) \qquad (\text{Corollary \ref{cor:exp-r_s}}) \tag*{\qed}
	\end{align*}
\end{proof}

\begin{proof}[of Theorem \ref{thm:support-product}]
	\begin{align*}
		\mathbb P(\Supp{(\Uv \Vv)} \neq EF) &= \mathbb P(\bigcup_{\phi \in EF^*\setminus\{0\}} \Supp{(\Uv \Vv)} \subset \ker(\phi))\\
		&\leq \sum_{\phi \in EF^*\setminus\{0\}} \mathbb P(\Supp{(\Uv \Vv)} \subset \ker(\phi))\\
		&\leq \sum_{\phi \in EF^*\setminus\{0\}} \Expec(2^{-n_2R_1}) \qquad (\textrm{Lemma \ref{lem:maj-kerphi}})\\
		&\leq 2^{rd} \, \Expec(2^{-n_2R_1})\\
		&\leq n_1 2^{rd-n_1n_2} \qquad (\textrm{Lemma \ref{lem:exp-r_1}}) \tag*{\qed}
	\end{align*}
\end{proof}

%% file: security.tex
\section{Security}\label{sec:secu}

\subsection{Definitions}

We define the $\INDCPA$-security of a KEM formally via the following experiment, where $\Encap_0$ returns a valid pair $c^*,K^*$, and $\Encap_1$ returns a valid $c^*$ and a random $K^*$.\\[2mm]

\begin{minipage}{\textwidth}
	\textit{Indistinguishability under Chosen Plaintext Attack}: This notion states that an adversary should not be able to efficiently guess which key is encapsulated.
\end{minipage}

\begin{center}
\fbox{
	\begin{minipage}{0.4\textwidth}
		$\Exp_{\E,\A}^{\ind-b}(\seck)$
		\begin{enumerate}
			\item $\param \sets \Setup(1^\seck)$
			\item $(\pk,\sk) \sets \KeyGen(\param)$
			\item $(c^*,K^*) \sets \Encap_b(\pk)$
			\item $b' \sets \A(\GUESS:c^*,K^*)$
			\item \comreturn $b'$
		\end{enumerate}
		
	\end{minipage}
}
\end{center}

\medskip

\begin{definition}[$\INDCPA$ Security]
	A key encapsulation scheme KEM is $\INDCPA$-secure if for
	every PPT (probabilistic polynomial time) adversary $\advA$,
	we have that
	{\small \[\AdvINDCPA{\KEM}{\advA}:=|\Pr[\INDCPAreal^\advA\Rightarrow 1]-\Pr[\INDCPArand^\advA \Rightarrow 1]|\] is negligible.}
\end{definition}

\subsection{IND-CPA proof}

\subsubsection{Unstructured LRPC-MS}

\begin{theorem}
Under the hardness of the $\LRPC$ (Problem \ref{prob:IndLRPC}) and $\DRSL_{k,n,r,\ell}$ (Problem \ref{prob:RSL}) problems, the KEM presented in section \ref{subsec:descriptionKEM} is indistinguishable against Chosen Plaintext Attack in the Random Oracle Model.
\end{theorem}

\begin{proof}
We are going to proceed in a sequence of games. The simulator first starts from the real scheme. First we replace the public key matrix by a random matrix, and then we use the ROM to solve Rank Support Learning.\\

We start from the normal game $G_0$:
We generate the public key $\Hv$ honestly, as well as $E$, and $\Cv$.
\begin{itemize}
	\item 
	In game $G_1$, we now replace $\Hv$ by a random matrix, the rest is identical to the previous game. From an adversary point of view, the only difference is the distribution of $\Hv$, which is either generated at random, or the systematic form of a low weight parity matrix. This is exactly the \emph{LRPC codes decisional} problem, hence \[\Adv_{\mA}^{G_0} \leq \Adv_{\mA}^{G_1} + \Adv_{\mA}^{\LRPC}.\]
	
	\item In game $G_2$, we now proceed as earlier except we receive $\Hv, \Cv$ from a Rank Support Learning challenger. After sending $\Cv$ to the adversary, we monitor the adversary queries to the Random Oracle, and pick a random one that we forward as our simulator answer to the $\DRSL_{k,n,r,\ell}$ problem. Either the adversary was able to predict the random oracle output, or with probably $1/q_G$, we picked the query associated with the support $E$ (by $q_G$ we denote the number of queries to the random oracle $G$), hence \[\Adv_{\mA}^{G_1} \leq 2^{-\lambda} + 1/q_G \cdot \Adv_{\mA}^{\DRSL}\] which leads to the conclusion.
	
\end{itemize}
\end{proof}

\subsubsection{Unstructured LRPC-MS with extended decoding (from Section \ref{subsec:reducing-m})}

Compared to the previous scheme, an attacker of the improved version of LRPC-MS knows a hash $G'(E)$ of the error support $E$.

In the random oracle model, $G'(E)$ is a value indistinguishable from random that gives no information on the shared secret $E$ so the IND-CPA proof is unchanged.

In practice, the knowledge of $G'(E)$ gives the attacker a way to quickly verify if a guessed subspace $E'$ is the correct error support by computing $G'(E')$. However in the general $\RSL$ problem if an attacker guesses correctly the support $E'$ he also gets a way to quickly verify its guess $E'$ by checking the existence of a solution to the equation $\Cv = \Hv\Vv$ of unknown $\Vv \in E'^{n\times\ell}$, even without knowing the hash $G'(E)$. Hence knowing the hash $G'(E)$ can only give at most a theoretical gain of the cost of solving a linear system divided by the cost of computing a hash function.



Now, the best attacks on $\DRSL$ do not enumerate all possible candidates $E'$ with a given dimension $r$ (there would be too many of them) but rather try to guess a space $E$ of greater dimension which contains $E'$, and not $E'$ directly so that in practice knowing $G'(E)$ does not help for best known attacks. Hence in practice giving the hash $G'(E)$ does not alter the security of the scheme.

\subsubsection{Ideal LRPC-MS}

For the ideal version of our scheme, the security proof is exactly the same except that ideal versions of hard problems appear. The IND-CPA property follows immediately.

\begin{theorem}
Under the hardness of problems $\ILRPC$ and $\DIRSL_{k,n,r,\ell}$, the KEM presented in section \ref{subsec:descriptionKEM} is indistinguishable against Chosen Plaintext Attack in the Random Oracle Model.
\end{theorem}

The maximal value of $\ell$ for which $\DIRSL_{k,n,r,\ell}$ is hard is way lower than its non-ideal counterpart. Indeed, a single ideal syndrome can be expanded in $k$ traditional syndromes by performing ideal rotations. That is why the value of $\ell$ is lower in the parameter sets for the ideal version.

\subsubsection{IND-CCA2}
It would be possible to get an IND-CCA2 KEM by applying the same strategies than ROLLO-II~\cite{ABDGHRTZABBBO19}, transforming our construction into an IND-CPA PKE and then applying the HHK framework to get an IND-CCA2 KEM.

\subsection{Known attacks}

As an $\RSL$ challenge is an $\RSD$ challenge with multiple syndromes, it is possible to try to solve $\RSL$ in two ways, either take only one syndrome to build an $\RSD$ challenge, and attack $\RSD$, or attack $\RSL$ with all the information in the challenge. In order to define the parameters sets for our schemes we therefore have to consider the bests attacks against both $\RSD$ and $\RSL$. At last we recall a specific attack against the $\LRPC$ problem.

\smallskip

There are two main types of attacks for solving the generic $\RSD$ problem: combinatorial attacks and algebraic attacks. For cryptographic parameters the best attacks are usually the recent algebraic attacks, but it may also depend on parameters, sometimes combinatorial attacks can be better.

\paragraph{Combinatorial attacks against $\RSD$.}

The best combinatorial attacks for solving the $\RSD$ problem on a random $[n,k]$ code over $\Fqm$ for a rank weight $d$ as described in \cite{AGHT18_sv} have complexity (for $\omega$ the linear algebra exponent):

\begin{equation}
 \min((n-k)^{\omega} m^{\omega} q^{(d-1)(k+1)},(km)^\omega
        q^{d\lceil\frac{km}{n}\rceil - m})
\end{equation}

The first term of the $\min$ typically corresponds to the case where $m \ge n$, the second term corresponds to the case where $m \le n$, but still it can happen that this term is better than the first one, when $m \ge n$ but close to $n$.  A detailed description of the complexity of the second term is given in \cite{AGHT18_sv}.

\paragraph{Algebraic attacks against $\RSD$.}

The general idea of algebraic attacks is to rewrite an $\RSD$ instance as a system of multivariate polynomial equations and to find a solution to this system.

For a long time, algebraic attacks were less efficient than combinatorial ones. Recent results improved the understanding of these attacks. The best algebraic attacks against $\RSD$ can be found in~\cite{bardet2020improvements} and have complexity (for $\omega$ the linear algebra exponent):
\begin{equation}	
q^{ar} m \binom{n-k-1}{r} \binom{n-a}{r}^{\omega-1}
\end{equation}
operations in $\Fq$. $a$ is defined as the smallest integer such that the condition $m \binom{n-k-1}{r}~\geq~\binom{n-a}{r} - 1$ is fulfilled.

\paragraph{On the security of the $\RSL$ problem.}

The difficulty of solving an instance of the $\RSL_{n, k, r, N}$ problem depends on the number $N$ of samples. Clearly, for $N=1$, the $\RSL$ problem is exactly the $\RSD$ problem with parameters $(n, k, r)$, which is probabilistically reduced to the NP-hard syndrome decoding problem in the Hamming metric in \cite{GZ14}. When $N \geqslant nr$, the $\RSL$ problem is reduced to linear algebra, as stated in \cite{GHPT17a_sv} where this problem was first introduced.

This raises the question of the security of the $\RSL$ problem in the case $1 < N < nr$. In \cite{GHPT17a_sv} the authors relate this problem to the one of finding a codeword of rank $r$ in a code of same length and dimension containing $q^N$ words of this weight, and conjecture that the complexity of finding such a codeword gets reduced by at most a factor $q^N$ compared to the case $N=1$. They also observe that in practice, the complexity gain seems lower, likely due to the fact that said codewords are deeply correlated.

There have been recent improvements on the complexity of the $\RSL$ problem. In \cite{DT18b} the authors show that the condition $N \leqslant kr$ should be met in order to avoid a subexponential attack. We chose our parameters to fulfill this condition.

In \cite{BB21}, the authors propose to solve the $\RSL$ problem in the case $N \leq kr$ using an algebraic approach.
Our parameters -- in particular the number $\ell$ of syndromes -- are chosen so as to resist these recent algebraic attacks.



\paragraph{On the security of the $\LRPC$ problem.}

Given $\Hv \in \Fqm^{(n-k)\times k}$ such that $(\Iv_{n-k}|\Hv)$ is the
parity-check matrix of a code $\C$, the problem of distinguishing LRPC
codes is to decide whether $\C$ is a random code or an LRPC code.

The best known attack against this problem for almost ten years (\cite{GRSZ14}) consists 
in using the underlying homogeneous structure of the LRPC code
to find a codeword of weight $d$ in a $[n-\lfloor\frac{n-k}{d}\rfloor, n-k-\lfloor\frac{n-k}{d}\rfloor]_{q^m}$ subcode
$\C'$ of the dual code $\C^\perp$ generated by $(\Iv_{n-k}|\Hv)$ rather than a codeword of weight $d$ in the 
$\C^\perp$ $[n,n-k]$ code. Then one can consider the previously described algebraic or combinatorial attacks for this slightly smaller code (but for the same weight $d$).

%% file: parameters.tex
\section{Parameters and performance}

\subsection{Parameters}

\paragraph{Choice of parameters.} In Section \ref{sec:secu}, the security of the protocol is reduced to the
$\LRPC$ and $\DRSL$ problems (or their ideal variants). The
best known attacks on these problems are thus used to define our parameters.
We also chose our parameters in order to have the Decoding Failure Rate (DFR) below or very close to $2^{-\lambda}$, where $\lambda$ is the security parameter, using Proposition \ref{prop:DFR} (or \ref{prop:DFR-extended} when using the extended algorithm). We only considered parameters with $k \geq \ell$ as required by these propositions.
We do not suggest parameters for LRPC-xMS-192 because the only possible prime value for $m$ below 151 is 149, which makes the gain too small to be considered.

\paragraph{Size of parameters.} One may use seeds to represent the random data in order to decrease
the keysize. We use the NIST seed expander with 40 bytes long seeds.

The public key $\mathsf{pk}$ is composed of a matrix of size $(n-k) \times n$ in a systematic form, so its size is $\Ceiling{\frac{k(n-k)m}{8}}$ bytes. The size is reduced to $\Ceiling{\frac{(n-k)m}{8}}$ bytes in the ideal case.
The secret key $\mathsf{sk}$ is composed of two random matrices that can be generated from a seed, so its size is 40 bytes.
The ciphertext $\mathsf{ct}$ is composed of a matrix of size $(n-k) \times \ell$, so its size is $\Ceiling{\frac{(n-k)\ell m}{8}}$ bytes. When using the extended $\RSR$ algorithm (denoted $\xRSR$), $64$ bytes are added to transmit a hash of $E$.
The shared secret $\mathsf{ss}$ is composed of $K = G(E)$, so its size is 64 bytes.

Parameters are given in Table~\ref{tab:params LRPC-MS}. The "structure" column indicates whether this parameter uses unstructured (random) matrices or ideal ones, and the "decoding" column indicates which decoding algorithm is used ($\RSR$ or $\xRSR$). The number indicated in the "DFR" column is actually $-\log_2$(DFR).

\begin{table}[!ht]
	\centering
	\begin{tabular}{|c|c|c|c|c|c|c|c|c|c|c|c|c|c|}
		\hline
		Instance & Decoding & $q$ & $n$ & $k$ & $m$ & $r$ & $d$ & $\ell$ & Security & DFR & $\pk$ size & $\ct$ size & $\pk + \ct$\\
		\hline
		LRPC-MS-128 & $\RSR$ & 2 & 34 & 17 & 113 & 9 & 10 & 13 & 128 & 126 & 4,083 & 3,122 & 7,205\\ \hline
		LRPC-xMS-128 & $\xRSR$ & 2 & 34 & 17 & 107 & 9 & 10 & 13 & 128 & 127 & 3,866 & 3,020 & 6,886\\ \hline
		\hline
		LRPC-MS-192 & $\RSR$ & 2 & 42 & 21 & 151 & 11 & 11 & 15 & 192 & 190 & 8,324 & 5,946 & 14,270\\ \hline
	\end{tabular}
	\caption{\label{tab:params LRPC-MS}Parameters for our unstructured LRPC-MS cryptosystem. The security is expressed in bits and sizes are expressed in bytes.}
        \vspace{-1em}
\end{table}

\begin{table}[!ht]
	\centering
	\begin{tabular}{|c|c|c|c|c|c|c|c|c|c|c|c|c|c|}
		\hline
		Instance & Decoding & $q$ & $n$ & $k$ & $m$ & $r$ & $d$ & $\ell$ & Security & DFR & $\pk$ size & $\ct$ size & $\pk + \ct$\\
		\hline
		ILRPC-MS-128 & $\RSR$ & 2 & 94 & 47 & 83 & 7 & 8 & 4 & 128 & 126 & 488 & 1,951 & 2,439\\ \hline
		ILRPC-xMS-128 & $\xRSR$ & 2 & 94 & 47 & 73 & 7 & 8 & 4 & 128 & 126 & 429 & 1,780 & 2,209\\ \hline
		\hline
		ILRPC-MS-192 & $\RSR$ & 2 & 178 & 89 & 109 & 9 & 8 & 3 & 192 & 189 & 1,213 & 3,638 & 4,851\\ \hline
		ILRPC-xMS-192 & $\xRSR$ & 2 & 178 & 89 & 97 & 9 & 8 & 3 & 192 & 189 & 1,080 & 3,238 & 4,318\\ \hline
	\end{tabular}
	\caption{\label{tab:params ILRPC-MS}Parameters for our ideal LRPC-MS cryptosystem. The security is expressed in bits and sizes are expressed in bytes.}
        \vspace{-2em}
\end{table}

\subsubsection{Comparison with other unstructured cryptosystems}
We compare our cryptosystem to other structured and unstructured proposals. Our comparison metric is the usual TLS-oriented communication size (public key + ciphertext).\\

For Loong.CCAKEM \cite{wang2019loong}, we consider only the third set of parameters since the other sets of parameters have an error weight below 6 and thus are vulnerable to algebraic attacks. For Loidreau cryptosystem, we consider the parameters presented in the conclusion of \cite{pham2021etude} which take into account the recent improvements on algebraic attacks. For both cryptosystems mentioned in this paragraph, parameters were not available (N/A) for 192 bits of security.

\begin{table*}[h!]
	\centering
	\begin{tabular}{|l|c|c|}  
		\hline
		Instance & 128 bits & 192 bits\\
		\hline
		LRPC-xMS & 6,886 & N/A\\
		\hline
		LRPC-MS & 7,205 & 14,270\\
		\hline
		Loong.CCAKEM-III & 18,522 & N/A\\
		\hline
		FrodoKEM &  19,336  & 31,376 \\
		\hline
		Loidreau cryptosystem &  36,300  & N/A\\
		\hline
		Classic McEliece & 261,248 & 524,348\\
		\hline
	\end{tabular}
	\caption{\label{tab:comparison-unstructured}Comparison of sizes of unstructured post-quantum KEMs. The sizes represent the sum of public key and ciphertext expressed in bytes.
	}
        \vspace{-2em}
\end{table*}\begin{table*}[h!]
	\centering
	\begin{tabular}{|l|c|c|}  
		\hline
		Instance & 128 bits & 192 bits\\
		\hline
		ILRPC-xMS & 2,209 & 4,318\\
		\hline
		ILRPC-MS & 2,439 & 4,851\\
		\hline
		BIKE & 3,113 & 6,197\\
		\hline
		HQC &  6,730 & 13,548\\
		\hline
	\end{tabular}
	\caption{\label{tab:comparison-structured}Comparison of sizes of structured code-based KEMs. The sizes represent the sum of public key and ciphertext expressed in bytes.
	}
        \vspace{-2em}
\end{table*}

\subsection{Performance}

This section provides indicative performance measurements of an implementation of some of the LRPC-MS cryptosystem parameters. Benchmarks were realized on an Intel\textsuperscript{\textregistered{}} Core\texttrademark{} i7-11850H CPU by averaging $1000$ executions. 

\begin{table}[h!]
\centering
\begin{tabular}{| l || c | c | c |}
\hline
	Instance  & KeyGen & Encap & Decap \\ \hline\hline
	LRPC-MS-128   & 383 & 137 & 3,195 \\ \hline
	LRPC-xMS-128   & 681 & 200 & 4,970 \\ \hline
	ILRPC-MS-128   & 214 & 107 & 1,213 \\ \hline
\end{tabular}
\caption{\label{tab:perfs_ref_cycles}Performances of our LRPC-MS cryptosystems in thousands of CPU cycles.}
\vspace{-2em}
\end{table}

As for other code-based schemes, the decapsulation algorithm has a higher computational cost than key generation and encapsulation. Note however, that  
our implementation does not yet benefit from the techniques of \cite{chou2022constant}. These techniques improved the decapsulation performance by a factor 15 (for 128 bits of security) with respect to the existing (and simpler to adapt) implementations we used as a basis for our benchmarking.

%% file: generalization_ouroboros.tex
\section{Application of the approach to other LRPC code-based cryptosystems}

The multi-syndrome approach could be applied to another rank-based KEM called \textbf{Ouroboros}.

\paragraph{Background on Ouroboros}

Ouroboros is a rank-based KEM that was first introduced in \cite{DGZ17}. It was renamed ROLLO-III for the second round of the NIST standaradization project \cite{ABDGHRTZABBBO19}.
Ouroboros and ROLLO-III use ideal LRPC codes and the RSR algorithm.
But unlike ROLLO-I and ROLLO-II, the security of this KEM does not depend on the indistinguishability of ideal LRPC codes but only on the underlying decoding problem.
This weaker security assumption comes at the price of a keysize slightly larger and a ciphertext size two times larger than for ROLLO-I.

\paragraph{Informal description of Ouroboros with multiple syndromes}

Sending multiple syndromes in the ciphertext opens the way for an efficient Ouroboros-like scheme without an ideal structure. The security of this KEM would thus only rely on the hardness of the (unstructured) $\DRSL$ problem. The hardness of the $\LRPC$ distinguishing problem does need not be assumed. We give an informal description of such a scheme in Figure \ref{fig:ouroboros}.\\

\begin{figure*}[bt]
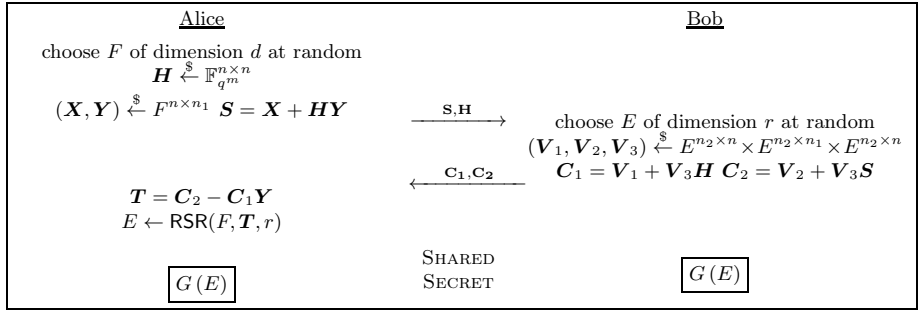

	
\resizebox{\textwidth}{!}{
\fbox{\hspace{-.0cm}%
\begin{minipage}{1.15\textwidth}
\begin{minipage}{.4225\textwidth}
	\parbox{\textwidth}{\centering
		\underline{Alice}}
\end{minipage} 
\begin{minipage}{.1125\textwidth}
	\parbox{\textwidth}{}
\end{minipage} ~
\begin{minipage}{.4225\textwidth}
	\parbox{\textwidth}{\centering
		\underline{Bob}}
\end{minipage} \\[2mm]
\begin{minipage}{.4225\textwidth}
	\parbox{\textwidth}{\centering%
		choose $F$ of dimension $d$ at random\\
		$\Hv \rand \Fqm^{n\times n}$\\
		$(\Xv, \Yv) \rand F^{n\times n_1}$
		$\Sv = \Xv + \Hv\Yv$\\[10mm]
		$\Tv = \Cv_2 - \Cv_1\Yv$\\
		$E \leftarrow \RSR(F, \Tv, r)$\\
		~\\[2mm]
		\fbox{$G\left(E\right)$}%
	}%
\end{minipage} ~
\begin{minipage}{.1125\textwidth}
	\parbox{\textwidth}{\centering%
		~\\[7mm]
		$\xrightarrow{~~~~\mathbf{S, H}~~~~}$\\[6mm]
		$\xleftarrow{~~~~\mathbf{C_1, C_2}~~~~}$\\[8mm]
		
		\textsc{Shared Secret}%
	}
\end{minipage} ~
\begin{minipage}{.4225\textwidth}
	\parbox{\textwidth}{\centering%
		~\\[8mm]
		choose $E$ of dimension $r$ at random\\
		$(\Vv_1, \Vv_2, \Vv_3) \rand E^{n_2\times n} \times E^{n_2\times n_1} \times E^{n_2\times n}$ \\
		$\Cv_1 = \Vv_1 + \Vv_3\Hv$
		$\Cv_2 = \Vv_2 + \Vv_3\Sv$\\[12mm]
		\fbox{$G\left(E\right)$}%
	}
\end{minipage}
\end{minipage}
}%
}
\caption{\label{fig:ouroboros}Informal description of Ouroboros with multiple syndromes.}
\end{figure*}

\paragraph{Expected performance}

The objective of this section is to give a general idea of the generalization but it is not the central argument of this article. Therefore we do not give precise parameters nor a complete description of this new key encapsulation mechanism. One can expect that the sum of public key and ciphertext sizes would be around 18 kB. This scheme can thus be a relevant alternative to the one presented in Section \ref{subsec:descriptionKEM}, if one wants to remove the hypothesis on the hardness of the LRPC distinguishing problem.